\documentclass{eptcs}
\providecommand{\event}{DICE 2010} 
\providecommand{\volume}{??}
\providecommand{\anno}{2010}
\providecommand{\firstpage}{1}
\providecommand{\eid}{??}
\usepackage{breakurl}        

\usepackage{color}
\usepackage[latin1]{inputenc}
\usepackage{amsmath}
\usepackage{amssymb}
\usepackage{amsthm}
\usepackage{graphicx} 
\usepackage{prooftree}
\usepackage{proof}

\newcommand{\vx}{\overline{x}}
\newcommand{\vy}{\overline{y}}

\newcommand{\vq}{\overline{q}}
\newcommand{\vm}{\overline{m}}
\newcommand{\vu}{\overline{u}}
\newcommand{\vv}{\overline{v}}
\newcommand{\vs}{\overline{s}}

\newcommand{\vC}{\overline{C}}
\newcommand{\vM}{\overline{M}}
\newcommand{\Val}[1]{[\![ #1 ]\!]}
\newcommand{\N} {\mathbb{N}}

\newtheorem{theorem}{Theorem}[section]
\newtheorem{lemma}[theorem]{Lemma}
\newtheorem{corollary}[theorem]{Corollary}
\newtheorem{prop}[theorem]{Proposition}
\theoremstyle{definition}
\newtheorem{definition}[theorem]{Definition}

\newtheorem{remark}[theorem]{Remark}

\theoremstyle{plain}  \newtheorem{theonum}{Theorem}%

\newenvironment{theorm_num}[1]
  {\begin{theonum}[{\bf #1}] }{\end{theonum}}

\def\Arr{{{\rightarrow}\mskip-11mu{\rightarrow}}}

\newcommand{\church}[1]{\mathsf{#1}^\bullet}
\newcommand{\scott}[1]{\mathsf{#1}^\circ}
\newcommand{\op}[1]{\mathbf{#1}}
\newcommand{\class}[1]{\textsc{\bf{#1}}}
\newcommand{\DIAL}{\mathbf{DIAL}_{lin}}
\newcommand{\LM}{\mathbf{1\lambda^p(W)}}

\author{Alo\"is Brunel
\institute{ENS Lyon}
\email{alois.brunel@ens-lyon.org}
\and
Kazushige Terui
\institute{RIMS, Kyoto University}
\email{terui@kurims.kyoto-u.ac.jp}
}

\title{Church $\Rightarrow$ Scott $=$ Ptime: \\
 an application of resource sensitive
realizability}
\def\titlerunning{Church $\Rightarrow$ Scott $=$ Ptime: an application of resource sensitive
realizability}
\def\authorrunning{A. Brunel \& K. Terui}

\begin{document}

\maketitle
\begin{abstract}
We introduce a variant of linear logic with second order quantifiers 
and type fixpoints, both restricted to purely linear formulas.
The Church encodings of  binary words are typed by 
a standard non-linear type `Church,'
while the Scott encodings (purely linear representations of words) are
by a linear type `Scott.'
We give a characterization of polynomial 
time functions, which is derived from (Leivant and Marion 93):
a function is computable in polynomial time if and only if 
it can be represented by a term of type 
Church $\Rightarrow$ Scott.

To prove soundness, we employ a resource sensitive 
realizability technique developed by Hofmann and Dal Lago.
\end{abstract}

\section{Introduction}
\label{introduction}

The field of implicit computational complexity aims to provide abstract,
qualitative, machine-independent characterizations of complexity classes
such as polynomial time and polynomial space functions. 
Along its development, two crucial factors for bounding 
complexity of programs have been identified:

\begin{description}
\item[Linearity:] In the higher order setting,
non-linear use of function variables often causes
an exponential growth of execution time. Hence 
a natural approach is to restrict use of 
higher order variables, often using types, in order to capture
the desired complexity classes. Examples are
light linear/affine logics \cite{girard1998light,AspertiRoversi:00}, 
their variant 
dual light affine logic \cite{baillot2009light}, 
soft linear logic \cite{Lafont2004}, 
and mixtures of linear higher order types with safe recursion
(eg., \cite{Bellantoni2000}, 
\cite{DBLP:journals/apal/Hofmann00}).
These logics all capture polynomial time functions, while 
there are also systems corresponding to polynomial space
\cite{DBLP:conf/popl/GaboardiMR08} and elementary 
functions \cite{girard1998light}.
\item[Data tiering:] Another source of exponential explosion lies in 
nested use of recursion, as observed by 
\cite{bellantoni1992new,DBLP:conf/lics/Leivant91}. 
Hence one naturally
restricts the structure of primitive recursive programs by data tiering.
This approach is most extensively pursued 
by a series of papers by
Leivant and Marion on 
tiered recursion (ramified recurrence)
\cite{leivant1993stratified,leivant1993lambda,Leivant94c,leivant1995ramified,leivant1997ramified}.
In tiered recursion, one has a countable number of copies of 
the binary word algebra, distinguished by tiers. Then 
a bad nesting of primitive recursion is avoided by requiring that 
the output of the defined function has a lower tier 
than the variable 
it recurses on. 
 	\end{description}

\noindent 	
{\bf Data tiering and higher order functionals.}
Along their development of ramified recurrence,
Leivant and Marion have made an interesting observation in
\cite{leivant1993lambda},
which reveals an intimate relationship between data tiering and higher order
functionals. They consider a simply typed 
$\lambda$-calculus over a first order word algebra, 
called $\LM$. The system is inherently equipped with two ``tiers'':
the first order word algebra (of base type $o$) as
the lower ``tier,'' and 
the Church encodings of words 
(of higher order type $(\tau\rightarrow \tau)^2 \rightarrow 
\tau\rightarrow \tau$)
as the higher one. First order words are just bit strings, while
Church words internalize the iteration scheme. 
Due to this inherent tiering, the programs from 
the Church words to the first order words capture
the polynomial time functions.

What their work reveals is a rather logical nature of tiers; 
in the end, tiering is nothing but the distinction between first order and
higher order data. It is then natural to go one step further towards
the logical direction, by replacing the first order algebra with the
\emph{linear} lambda terms, and by identifying the higher order data
with \emph{non-linear} terms. Our intuition is backed up by the fact that
the linear encoding of words, 
often attributed to Scott (cf.\ \cite{n-types}), 
behaves very similarly to the
first order words; for instance, they admit constant time
successor, predecessor and discriminator, while they are not enhanced with
the power of iteration in their own. 

To identify the set of Scott words, 
it is useful to introduce a type system with linearity
and type fixpoints. 
We therefore introduce a variant of linear logic, called $\DIAL$,
as a typing system for the pure $\lambda$-terms.
This system distinguishes non-linear and linear arrows
and has second order quantifiers and type fixpoints, both restricted
to linear types.
Morally, the base type of $\LM$ corresponds to the hereditarily
linear formulas of $\DIAL$,
and the higher types of $\LM$ to 
the non-linear formulas.
We then characterize the class of polynomial time functions 
as those represented by
terms of type: `Church' (nonlinear words) $\Rightarrow$ `Scott' 
(linear words).
The two types for
binary words play the role of the two tiers. 
Our work thus exhibits a connection between the two factors controlling 
complexity:
linearity and data tiering.

\medskip
\noindent{\bf Resource sensitive realizability.}
Following some preceding works \cite{hofmann1997application,DBLP:journals/apal/Hofmann00,DBLP:journals/iandc/Hofmann03},
Dal Lago and Hofmann have introduced in \cite{dal2005quantitative} 
a realizability semantics which is useful to reason about
the complexity bounds for various systems uniformly. 
In their framework,
the realizers are pure $\lambda$-terms (values, to be more precise) under 
the weak call-by-value semantics, and they come equipped with 
the resource bounds expressed by elements of a resource monoid. 
Various systems are then dealt with 
by choosing a suitable resource monoid, while the basic
realizability constructions are unchanged.
This framework has offered new and uniform proofs of the soundness theorems
for LAL, EAL, LFPL, SAL and BLL
with respect to the associated complexity classes 
\cite{dal2005quantitative,DBLP:conf/tlca/LagoH09}.  

We here apply their technique to prove that 
all terms of type Church $\Rightarrow$ Scott in the system  
$\DIAL$ are polytime. 
The main novelty is that we build a suitable (partial) resource monoid
based on higher order polynomials. Also, we do not require
that realizers are values. This allows us to directly infer
the complexity bounds of arbitrary $\lambda$ terms
(not restricted to values).

\medskip
\noindent{\bf Outline.}
Section \ref{section_2} introduces the system $\DIAL$ and 
states the main results. Section \ref{section_3} introduces 
the realizability semantics and proves the adequacy theorem.
Section \ref{section_4} applies these tools to derive the 
soundness theorem.
Section \ref{conclusion} concludes this work. 
 
\section{System $\DIAL$}

In this section, we recall the weak call-by-value $\lambda$-calculus
with the time cost measure of \cite{dal2008weak}, and then 
introduce the type system $\DIAL$ 
derived from second order affine linear logic with type fixpoints. 
The system emulates the two tiers of 
$\LM$
by distinguishing linear and non-linear types.
\label{section_2}
\subsection{Weak call-by-value lambda calculus with time measure}
\label{computational}

We assume that a set of variables $x,y,z,\dots$ are given.
As usual, the \emph{$\lambda$-terms} $t, u$ are defined by the grammar:
$t,u ::= x\,\,|\,\, \lambda x.t \,\,|\,\, tu$.
The set of $\lambda$-terms is denoted by $\Lambda$.
Terms of the form $x$ or $\lambda x.t$ are called \emph{values}.
We denote by $FV(t)$ the set of the free variables of $t$
and by $\Val{t}_\beta$ the $\beta$-normal form of $t$.
The \textit{size} $|t|$ of a term $t$ is defined by:
$$
|x|  =  1, \qquad
|\lambda x. t| =  |t| + 1, \qquad
|tu|  =  |t| + |u|.
$$

As with \cite{dal2008weak},
we adopt the \textit{weak call-by-value} reduction strategy, which is defined by: 
$$
\infer{(\lambda x.t)v \rightarrow t[v/x]}{}
\qquad
\infer{t_1 u \rightarrow t_2 u}{t_1 \rightarrow t_2}
\qquad
\infer{u t_1  \rightarrow u t_2}{t_1 \rightarrow t_2}
$$
where $v$ denotes a value. We write $t\Downarrow$ if
$t$ evaluates to a value $v$: $t \rightarrow^* v$. 
The value $v$ is unique whenever $t \rightarrow^* v$, so we write 
$\Val{t} = v$.
It should not be confused with the $\beta$-normal form $\Val{t}_\beta$
of $t$.

The cost of evaluation is specified by 
a ternary relation 
$t \stackrel{n}{\Arr} u$, meaning that $t$ reduces to $u$ with cost 
$n$, defined as follows:
$$
\infer{t \stackrel{0}{\Arr} t}{}
\qquad
\infer{t \stackrel{n}{\Arr} u}{
t \rightarrow u & n = max\{|u| - |t|,1\}}
\qquad
\infer{s \stackrel{n+m}{\Arr} u}{
s \stackrel{n}{\Arr} t & t \stackrel{m}{\Arr} u}
$$
The definition takes into account 
the cost of duplications. In particular we have:
\begin{lemma}\label{l-cost}
Suppose that 
$(\lambda x.t)v \stackrel{n}{\Arr} t[v/x]$
and $x$ occurs $c$ times in $t$.
Then $n=1$ if $c\leq 1$, and
$n \leq (c-1)|v|$ if $c\geq 2$.
\end{lemma}

\begin{proof} In the first case,
$|t[v/x]| < |(\lambda x.t)v|$. In the second case,
$|t[v/x]| - |(\lambda x.t)v| \leq |t| + c|v| - (|t| + 1 + |v|)
\leq (c-1)|v|$.
\end{proof}


A distinctive feature of 
the above cost model is that the cost $n$ is unique:
$t \stackrel{n}{\Arr} v$ and 
$t \stackrel{m}{\Arr} v$ imply $n = m$ \cite{dal2008weak}. 
So we may define $Time (t) = n$ without ambiguity
($Time(t)$ is undefined if $t \not\Downarrow$). 
Finally, let $TS(t) = Time(t) + |t|$.
(It should be noticed that $TS(t)$ is denoted as
$Time(t)$ in \cite{dal2008weak}; our notation is due to
\cite{dal-semantic}.) 

It is proved in \cite{dal2008weak} that this cost model is invariant, 
which means that $\lambda$-calculus and Turing machines 
simulate each other with
a polynomial time overhead. In particular, we have:

\begin{theorem}\label{t-turing}
There exists a Turing machine $M_{eval}$ with the following property:
given  a $\lambda$-term $t$ such that $t\Downarrow$
and  $TS(t) = Time(t) + |t| = n$,
$M_{eval}$ computes $\Val{t}$ in time $O(n^4)$.
\end{theorem}

The following facts (cf.\ \cite{dal2008weak}) will be useful below.

\begin{lemma}
\label{value_lemma} The following hold when $t\Downarrow$.
\begin{description}
\item[(size)] $|\Val{t}| \leq TS(t)$.
\item[(exchange)] If 
$t= (\lambda x_1 x_2.s)u_1 u_2$ and 
$t' = (\lambda x_2 x_1.s)u_2 u_1$,
then $TS(t') = TS(t)$.
\item[(contraction)] If
$t = (\lambda x_1 x_2. s) u u$ and 
$t' = (\lambda x. s[x/x_1, x/x_2]) u$, then
$TS(t') \leq TS(t)$.
\item[(weakening)] If
$t' = (\lambda x. t)u$, $x\not\in FV(t)$ and $u\Downarrow$, then 
$TS(t') = TS(t)+ TS(u)+2$.
\item[(concatenation)] If
$t = s_1 (( \lambda x.s_2)u)$ and
$t' = (\lambda x. s_1 s_2)u$ ($x \not\in FV(s_1)$), then
$TS(t') = TS(t)$.
\item[(identity)] If 
$t' = (\lambda x.x)t$, then 
$TS(t') = TS(t) + 3$.
\end{description}
\end{lemma}

\begin{proof} For (size),
it is sufficient to prove that 
if $t \rightarrow u$ then $TS(t)\geq TS(u)$.
If $|u| - |t| \geq 1$, we have 
$TS(t)= Time(t) + |t| = (|u|-|t| + Time(u)) + |t| = Time(u)+ |u| =TS(u)$.
Otherwise,  
$ TS(t) = (1 + Time(u)) + |t| \geq Time(u)+ |u| = TS(u).$\\
For (weakening), we have $t' 
\stackrel{n}{\Arr} (\lambda x.t)\Val{u}
\stackrel{1}{\Arr} t$ with $n= Time(u)$.
Hence 
$TS(t') = Time(t') + |t'| =
(n + 1 + Time(t)) + (|t| + |u| + 1)
= TS(t) + TS(u) + 2$.\\
For (identity), we have $t'
\stackrel{n}{\Arr} (\lambda x.x)\Val{t}
\stackrel{1}{\Arr} \Val{t}$ with $n= Time(t)$.
Hence
$TS(t') = Time(t') + |t'| = (n + 1) + (|t| + 2)
= TS(t) + 3$.\\
For (contraction), we have
$t \stackrel{2n}{\Arr} (\lambda x_1 x_2.s)\Val{u}\Val{u} 
\stackrel{m}{\Arr} s[\Val{u}/x_1,\Val{u}/x_2] = t_0$ and
$t' \stackrel{n}{\Arr} (\lambda x.s[x/x_1, x/x_2])\Val{u}$
$\stackrel{k}{\Arr}$ $t_0$. Consider the case when each of $x_1$ and $x_2$ 
occurs more than once in $s$. Then
$m = |t_0| - 
|(\lambda x_1 x_2.s)\Val{u}\Val{u}|$ and 
$k = |t_0| - |(\lambda x.s[x/x_1, x/x_2])\Val{u}|$.
Hence we have:
\begin{eqnarray*}
TS(t) & = &
2n + (|t_0| - |(\lambda x_1 x_2.s)\Val{u}\Val{u}|) + Time(t_0) + |t| \\
 & = & 
2n + TS(t_0) + 2|u| - 2|\Val{u}|;\\
TS(t') & = & 
n + (|t_0| - |(\lambda x.s[x/x_1, x/x_2])\Val{u}|) + Time(t_0) + |t'| \\
  & = & 
n + TS(t_0) + |u| - |\Val{u}|.
\end{eqnarray*}
By (size), we have $|\Val{u}| \leq TS(u) = n + |u|$, hence we conclude
$TS(t') \leq TS(t)$.
The calculation is similar when either $z_1$ or $z_2$ occurs at most once.

The equations for (exchange) and (concatenation) are easily verified.
\end{proof}

\subsection{The dual type system}

We now introduce the system $\DIAL$: the dual intuitionistic affine 
logic with linear quantifiers and type fixpoints. 
It is based on intuitionistic
linear logic with unrestricted weakening (thus ``linear'' actually 
means ``affine'').
It does not possess the $!$ connective but distinguishes 
linear and non-linear function spaces as in \cite{baillot2009light}. 
It has the second order quantifier and the type fixpoint operator,
but both are restricted to purely linear formulas.

Given a set of propositional variables $\alpha, \beta, \dots$,
the (general) \emph{formulas} $A, B, \dots$ and 
the \emph{linear formulas} $L, M, \dots$ are defined by the following
grammar:
$$ L,M ::= \alpha\,\,|\,\, \forall \alpha L\,\,|\,\,\mu\alpha L^{(*)}\,\,|\,\,L\multimap M, \qquad 
A,B ::= L\,\,|\,\, \forall \alpha A \,\,|\,\, L \multimap B\,\,|\,\, A\Rightarrow B.$$
$(*)$ : we add the condition that we can build $\mu \alpha L$ only if $\alpha$ 
occurs only positively in $L$. 
This is a common restriction that makes it easier 
to interpret fixpoint types in realizability semantics.

Thus the linear formulas are  the formulas that do not contain 
any $\Rightarrow$. 

We handle judgments of the form $\Gamma;\Delta \vdash t : A$, 
where $\Delta$ consists of assignments of the form $(x:L)$ 
with $L$ a linear formula, and
$\Gamma$ consists of $(x:A)$ 
with $A$ an arbitrary formula. We assume that variables in $\Gamma$ and
$\Delta$ are distinct.
The variables in $\Delta$ are 
intended to be affine linear: each of them occurs at most once in $t$, 
in contrast to 
those in $\Gamma$ which may have multiple occurrences.
The typing rules are defined in Figure \ref{typing_rules}. 
Notice that $L$ always denotes a linear formula. 

\begin{figure}[ht!]
\centering
\begin{tabular}{|cc|}  
  \hline
\prooftree
   \justifies  x : A; \vdash x : A
\thickness=0.08em
\shiftright 0em\using (ax1) \endprooftree
&
\prooftree
   \justifies  ; x : L \vdash x : L
\thickness=0.08em
\shiftright 0em\using (ax2) \endprooftree
\\
&\\
\prooftree
\Gamma;\Delta \vdash t : \mu \alpha L \justifies \Gamma;\Delta \vdash  t : L[\mu \alpha L/\alpha]
\thickness=0.08em
\shiftright 0em\using (\mu_e) \endprooftree 
 &
 
\prooftree
\Gamma; \Delta \vdash  t : L[\mu \alpha L/\alpha]  \justifies \Gamma; \Delta \vdash  t : \mu \alpha L
\thickness=0.08em
\shiftright 0em\using (\mu_i) \endprooftree \\
&\\
\prooftree
\Gamma ; \Delta \vdash t : A \quad\quad \alpha \notin FV(\Gamma;\Delta) \justifies \Gamma ; \Delta \vdash t : \forall \alpha A
\thickness=0.08em
\shiftright 0em\using (\forall_i) \endprooftree 
 &
 
\prooftree
\Gamma;\Delta \vdash  t : \forall \alpha A \justifies \Gamma; \Delta 
\vdash t : A[L/\alpha]
\thickness=0.08em
\shiftright 0em\using (\forall_e) \endprooftree \\ 
 &\\
\prooftree
\Gamma_1 ; \Delta \vdash t : A \Rightarrow B \quad\quad  \Gamma_2 ; \vdash u : A \justifies  \Gamma_1, \Gamma_2 ;\Delta \vdash tu : B
\thickness=0.08em
\shiftright 0em\using (\Rightarrow_e) \endprooftree 
& 
\prooftree
\Gamma, z : A; \Delta \vdash t : B  \justifies \Gamma;\Delta \vdash \lambda z\,t : A \Rightarrow B
\thickness=0.08em
\shiftright 0em\using (\Rightarrow_i) \endprooftree \\
&\\
\prooftree
\Gamma_1 ; \Delta_1 \vdash t : L \multimap B \quad\quad  \Gamma_2 ; \Delta_2 \vdash u : L \justifies \Gamma_1, \Gamma_2 ; \Delta_1, \Delta_2 \vdash tu : B
\thickness=0.08em
\shiftright 0em\using (\multimap_e) \endprooftree 
& 
\prooftree
\Gamma; \Delta, z : L \vdash t : B \justifies \Gamma;\Delta \vdash \lambda z\,t : L \multimap B
\thickness=0.08em
\shiftright 0em\using (\multimap_i) \endprooftree \\
& \\
\multicolumn{2}{|c|}{
\quad
\prooftree
 \Gamma, x : A, y: A; \Delta  \vdash t : B\justifies  \Gamma, z : A ; \Delta \vdash t[z/x, z/y] : B
\thickness=0.08em
\shiftright 0em\using (Contr) \endprooftree
\qquad
\prooftree
 \Gamma ; \Delta, x : L \vdash t : B\justifies  \Gamma, x : L ; \Delta \vdash t : B
\thickness=0.08em
\shiftright 0em\using (Derel) \endprooftree
\qquad
\prooftree
 \Gamma ; \Delta  \vdash t : B\justifies  \Gamma, \Gamma'  ; \Delta, \Delta' \vdash t : B
\thickness=0.08em
\shiftright 0em\using (Weak) \endprooftree\quad}\\
& \\
  \hline
\end{tabular}
\caption{Typing rules of $\DIAL$}
\label{typing_rules}
\end{figure}

We say that a term $t$ is \emph{of type} $A$ in $\DIAL$ if 
$\vdash t: A$ is derivable by the typing rules in Figure 
\ref{typing_rules}.
Below are some remarks.
\begin{itemize}
\item The intended meaning of judgment $\Gamma; \Delta \vdash t : A$ is
$!\Gamma^*, \Delta \vdash t: A^*$, where $\Gamma^*, 
A^*$ are translations into linear logic given by 
$(B \Rightarrow C)^* = !B^* \multimap C^*$. Hence the rule
(Contr) can be applied only to variables in $\Gamma$.

\item The $\Rightarrow_e$ rule implicitly performs the $!$ promotion on
$A$, so the judgment for $u$ should not contain a linear variable.
\item We only allow substitution 
of linear formulas for propositional variables (in rules $(\forall_e)$, $(\mu_i)$ and $(\mu_e)$). One can check that such a substitution in a formula always results in a formula. This restriction is strictly necessary, since 
the exponential function would be typed otherwise
(see below).

\item One unpleasant restriction is that the premise $L$ 
of a linear implication $L\multimap B$ has to be linear. 
It does not seem essential for complexity, 
but our realizability argument forces it.

\item The type system enjoys the subject reduction property 
with respect to the $\beta$-reduction.
\end{itemize}


\subsection{Church and Scott data types}

In $\DIAL$, data may be represented in two ways, either in the
Church style or in the Scott style. Figure \ref{f-encoding} illustrates
the two encodings for natural numbers $n$ and 
binary words $w \in \{0,1\}^*$, together with 
some basic functions defined on them.
In the definition of $\church{w}$, $w$ is assumed to be 
$i_1 \cdots i_n$ where each $i_k$ is either $0$ or $1$.

The first thing to be verified is the following:
\begin{prop}
For every term $t$ in $\beta$-normal form, 
$\vdash t : \church{N}$ if and only if 
$t$ is a Church numeral $\church{n}$ (or $\lambda x.x$, that is 
$\eta$-equivalent to $\church{1}$).
$\vdash t : \scott{N}$ if and only if 
$t$ is a Scott numeral $\scott{n}$.
Similarly for Church and Scott words.
\end{prop}

\begin{proof} The claim is standard for Church numerals. So let us 
focus on Scott numerals. The following derivations show that 
$\vdash \scott{n}: \scott{N}$ for every natural number $n$.
$$
\infer{\vdash \scott{0} : \scott{N}}{
\infer{\vdash \scott{0} : 
\forall\alpha.(\scott{N}\multimap \alpha)\multimap (\alpha\multimap \alpha)}{
\infer={\vdash \scott{0} : (\scott{N}\multimap \alpha)\multimap (\alpha\multimap \alpha)}{
\infer{;x: \scott{N}\multimap \alpha, y: \alpha \vdash y:\alpha}{;y: \alpha \vdash y:\alpha}}}}
\qquad
\infer{\vdash \scott{(n+1)} : \scott{N}}{
\infer{\vdash \scott{(n+1)} : 
\forall\alpha.(\scott{N}\multimap \alpha)\multimap (\alpha\multimap \alpha)}{
\infer={\vdash \scott{(n+1)} : (\scott{N}\multimap \alpha)\multimap (\alpha\multimap \alpha)}{
\infer{;x: \scott{N}\multimap \alpha, y: \alpha \vdash x\scott{n}:\alpha}{
\infer={;x: \scott{N}\multimap \alpha, y: \alpha \vdash 
x: \scott{N}\multimap \alpha}{} 
&
\vdash \scott{n}:\scott{N}}}}}
$$
For the other direction, we proceed by induction on the size of $t$.
Suppose that 
$\vdash t : \scott{N}$. Since
$t$ is in $\beta$-normal form,
the last part of the derivation must be necessarily 
of the form 
$$
\infer{\vdash \lambda xy. t_0: \scott{N}}{
\infer{\vdash \lambda xy. t_0: \forall \alpha.(\scott{N}\multimap \alpha)\multimap(\alpha\multimap\alpha)}{
\infer={\vdash \lambda xy. t_0: (\scott{N}\multimap \alpha)\multimap(\alpha\multimap\alpha)}{
;x:\scott{N}\multimap \alpha, y:\alpha\vdash t_0: \alpha}}}
$$
and $t= \lambda xy. t_0$. Since $t_0$ is no more an abstraction, 
it must be either $y$ or of the form $x t_1$ with $\vdash t_1 : \scott{N}$.
In the former case
 we have $t = \scott{0}$, while in the latter case we may apply the
induction hypothesis to obtain 
$t_1 = \scott{n}$ for some $n$. Hence $t = \lambda xy.x \scott{n}
= \scott{(n+1)}$.
\end{proof}

Let us come back to Figure \ref{f-encoding}.
As usual, Church numerals $\church{n}, \church{m}$ can be multiplied
by composition $\church{n} \circ \church{m} = \lambda f. \church{n}(\church{m}f)$.
This can be repeated arbitrary many but fixed times,
so we naturally obtain terms
$\church{mult}$ and $\church{mon}_n$ representing multiplication and
monomial $x\mapsto x^n$ of degree $n$. 
On the other hand,
it is not possible to encode exponentiation, since 
it requires of 
instantiation of $\alpha$ with a non-linear formula
such as  $\church{N}$, that is not allowed in $\DIAL$.

Turning on to the Scott numerals and words, observe that 
they are affine linear, and
admit constant time successor $\scott{succ}$ and predecessor
$\scott{pred}$
in contrast to Church.

Every finite set of cardinality $n$ can be represented by $\scott{B}_n$,
and the tensor product of two linear formulas by $L \otimes M$.
These allow us to \emph{linearly} 
represent the decomposer $\scott{dec}$, which works 
as follows: $\scott{dec}(\scott{iw}) = \scott{b}_i \otimes \scott{w}$
for $i\in \{0,1\}$
and $\scott{dec}(\scott{\epsilon}) = \scott{b}_2 \otimes \scott{\epsilon}$.

Given these building blocks, it is routine to encode 
the transition function of a Turing 
machine by a term of \emph{linear} type $L \multimap L$.
It can then be iterated by means of 
$\church{iter} :
 \church{N} \Rightarrow (L \multimap L) \Rightarrow
 (L \multimap L)$.
Combining it with $\church{mon}_n$ and other ``administrative'' operations,
we obtain an 
encoding of arbitrary polynomial time Turing machines.

\begin{figure}
$
\begin{array}{|rclrcl|}
\multicolumn{6}{l}{\mathbf{Church\ numerals\ and\ words:}} \\
\hline
\church{N} & \equiv & 
\forall  \alpha (\alpha\multimap\alpha)\Rightarrow (\alpha\multimap\alpha) 
&
\church{W} & \equiv &
\forall \alpha(\alpha\multimap\alpha)\Rightarrow (\alpha\multimap\alpha)\Rightarrow (\alpha\multimap\alpha) \\
\church{n} &=& \lambda fx.\underbrace{f(...f}_{n\,\, times}(x)...) &
\church{w} &=& \lambda f_0.\lambda f_1.\lambda x.f_{i_1}(f_{i_2}(...(f_{i_n}(x)...)))\\
	\church{mult} & \equiv & \lambda x y \lambda f.
x(yf) :   \church{N} \Rightarrow \church{N} \Rightarrow \church{N}
&
  \church{mon}_n & \equiv & \lambda x\lambda f.
\underbrace{x( \cdots (x}_{n\ times} f)\cdots ):
 \church{N} \Rightarrow \church{N}\\
\hline
\multicolumn{6}{l}{\mathbf{Scott\ numerals\ and\ words:}} \\
\hline
\scott{N} & \equiv & \mu\beta \forall \alpha  (\beta \multimap \alpha) 
\multimap (\alpha \multimap \alpha) &
\scott{W} & \equiv & \mu\beta \forall \alpha  (\beta \multimap \alpha) \multimap (\beta\multimap\alpha)  \multimap (\alpha \multimap \alpha) \\
\scott{0} &=& \lambda xy.y & 
\scott{\epsilon} &=& \lambda xyz.z \\
\scott{(n+1)} & = & \lambda xy.x(\scott{n}) &
\scott{(0w)}& = & \lambda xyz.x(\scott{w}) \\
 & & &  \scott{(1w)}& = & \lambda xyz.y(\scott{w}) \\
 \scott{succ} & = & \lambda z.\lambda xy.xz : \scott{N}\multimap \scott{N} &
 \scott{pred} & = & \lambda z.z(\lambda x.x)(\scott{0}): 
\scott{N}\multimap \scott{N}\\
\hline
\multicolumn{6}{l}{\mathbf{Finite\ sets\ and\ tensor\ product:}}\\
\hline
\scott{B}_n & \equiv & 
\forall \alpha. \underbrace{\alpha \multimap ... \alpha\multimap}_{n\,\, times}
 \alpha & 
L \otimes M & \equiv & \forall \alpha. 
(L\multimap M\multimap \alpha)\multimap \alpha \\
\scott{b}_i & \equiv & \lambda x_0 \cdots x_{n-1}. x_i 
&
t \otimes u & \equiv & \lambda x. xtu 
\quad (t: L, \ u: M)\\
\hline
\multicolumn{6}{l}{\mathbf{Decomposer\ and\ iteration:}}\\
\hline
\scott{dec} & = &
\multicolumn{4}{l|}{
\lambda z. z(\lambda y. \scott{b}_0 \otimes y)
(\lambda y. \scott{b}_1 \otimes y)(\scott{b}_2 \otimes \scott{\epsilon})
\ : \ \scott{W} \multimap \scott{B}_3 \otimes \scott{W}}\\
\church{iter} & = & 
\multicolumn{4}{l|}{
\lambda xfg. xfg \ : \ \church{N} \Rightarrow (L \multimap L) \Rightarrow
 (L \multimap L)}\\
\hline
\end{array}
$
\caption{Basic encodings}\label{f-encoding}
\end{figure}

\begin{theorem}[\class{FP}-completeness]
\label{completeness_theo}
For every polynomial time function $f:
\{0,1\}^* \rightarrow \{0,1\}^*$, there exists a 
$\lambda$-term $t_f$ of type $\church{W} \Rightarrow \scott{W}$ 
in $\DIAL$.
Given $w \in \{0,1\}^*$, we have
$\Val{t_f \church{w}}_\beta = \scott{f(w)}$.
\end{theorem}

A couple of remarks are in order.
\begin{itemize}
\item  Both Church and Scott numerals/words
can be generalized to 
lists, trees and their combinations. 
It is indeed an advantage of the polymorphic setting that 
there is a generic means to build various data types.
Moreover, we may consider
for instance the Church lists of Scott numerals. 

\item In view of the fact that our system is derived from 
$\LM$ of 
\cite{leivant1993lambda},
one may wonder whether it is possible to give a direct translation of
$\LM$ into $\DIAL$ for proving \class{FP}-completeness.
It is, however, not 
straightforward because 
$\LM$ is not sensitive to the distinction between 
linear and non-linear arrows, that is crucial for our system.
In particular, our Church numerals 
only allow 
iteration of linear functions $L\multimap L$, while 
the Church numerals of $\LM$ 
allow iteration of
non-linear
functions as well.
\end{itemize}

The rest of this paper is concerned with the converse
of Theorem \ref{completeness_theo}.
Namely, we prove:

\begin{theorem}[\class{FP}-soundness]
\label{soundness_theo}
For every $\lambda$-term $t$ of type $\church{W} \Rightarrow \scott{W}$, 
the associated function 
$f_t : \{0,1\}^* \rightarrow \{0,1\}^*$ defined by $f_t(w_1) = w_2$ 
$\Leftrightarrow$ $\Val{ t \church{w_1}}_\beta = \scott{w_2}$
 is a polynomial time function.
\end{theorem}

Altogether, these two theorems ensure 
that the terms of type $\church{W}\Rightarrow\scott{W}$ in $\DIAL$
 precisely capture the class \class{FP} 
of polynomial time functions.

\section{Resource sensitive realizability}
\label{section_3}

We now develop a resource sensitive realizability semantics
for $\DIAL$
inspired by \cite{dal-semantic}. It concerns with the 
realizability relation $t, p \Vdash_\eta A$, where $A$ is a formula
to be realized, $\eta$ is a valuation of propositional variables, 
and $t$ is 
a $\lambda$-term, called a \emph{realizer}, that embodies the 
computational content of a given proof. The second component
$p$ is a higher order (additive) polynomial, called a \emph{majorizer},
that imposes a resource bound on $t$. Since we do not intend our model
to be categorical,  we do not include
the denotation of $t$ in the realizability relation 
(in contrast to the length space of \cite{dal-semantic}).

We then show the adequacy theorem, ensuring that $\DIAL$ is sound
with respect to the realizability semantics. 
 
\subsection{Higher order polynomials}

We begin with the description of majorizers, namely higher order 
polynomials. Actually they are just monotone additive terms (without 
multiplication), but we nevertheless call them polynomials,
since they will indeed serve as 
polynomials bounding the runtime of realizers 
(see Theorem \ref{weak_sound}).
Using higher order polynomials rather than first order ones 
will allow us to capture the difference between linear and non-linear formulas. 

\begin{definition}[Higher order polynomials]
We consider simple types $\sigma, \tau, \dots$ 
defined by $\sigma :: = o\ |\ 
\tau \rightarrow \tau$, where $o$ is the only base type.
A \emph{higher order polynomial} $p$ 
is a $\lambda$-term built from constants $n : o$ (for every 
natural number $n$) and $+ : o \rightarrow o \rightarrow o$. More precisely,
given a set $V(\sigma)$ of variables for each simple type $\sigma$, 
they are built as follows:
$$
\infer{x : \sigma}{x\in V(\sigma)}
\quad\quad
\infer{pq: \tau}{p: \sigma\rightarrow \tau & q: \sigma}
\quad\quad
\infer{\lambda x.p : \sigma\rightarrow \tau}{x \in V(\sigma) & p: \tau}
\quad\quad 
\infer{n: o}{n \in \mathbb{N}}
\quad\quad
\infer{+: o\rightarrow o\rightarrow o}{}
$$
We denote by $\Pi$ the set of closed higher order polynomials.
\end{definition}

The role of higher order polynomials is to impose a static, quantitative bound
on realizers. Hence we identify them 
by $\alpha\beta\eta$-equivalence
and natural arithmetical equivalences. For instance, we identify 
$x+ y = y+ x$ and $2 + 3 = 5$. We often write $p(q_1, \dots, q_n)$ 
for $p q_1 \cdots q_n$.
If $p:o$ and $c\in \N$, we write 
$c p$ for $p + \cdots + p$ ($c$ times).

We extend addition to higher-order terms so that one can sum up 
two terms at least when one of the summands is of base type $o$.
Formally, 
let $\tau = \tau_1 \rightarrow ... \rightarrow \tau_k \rightarrow o$
and $p:\tau$. If $q:o$, 
we denote by $p+q$
the term $\lambda x_1\cdots x_k.(p(x_1,...,x_k) + q)$. 

We also define a lowering operator which brings a higher order term
down to a base type one. It will allow majorizers of higher order type
to bound concrete resources such as time and size.
$$
\begin{array}{rcll}
0_\tau & = & 
\lambda x_1 \cdots x_k.0, &
\mbox{where }
\tau = \tau_1 \rightarrow ... \rightarrow \tau_k \rightarrow o;
\\
\downarrow p & = & p 0_{\tau_1} \cdots 0_{\tau_k}, &
\mbox{where }
p: \tau = \tau_1 \rightarrow ... \rightarrow \tau_k \rightarrow o.
\end{array}
$$
Observe that $\downarrow p$ is a natural number if $p$ is a closed
higher order polynomial. Notice also that $\downarrow p = p$ if 
$p : o$.

Formulas of $\DIAL$ are mapped to types of higher order polynomials 
as follows:
$$
o(L)= o, \qquad
o(L\multimap A) = o(A), \qquad
o(A\Rightarrow B) = o(A) \rightarrow o(B),
\qquad o(\forall \alpha\, A) = o(A).
$$
Thus all linear formulas collapse to $o$, while non-linear formulas
retain the structure given by non-linear arrows.

\begin{remark}\label{r-monoid}
Consider $\mathcal{M}=(\Pi, +, \leq, D)$ 
where $p_1+p_2$ is a partial operation defined only 
when one of the $p_i$ is of type $o$, $p\leq q$ 
iff $\downarrow p \leq \downarrow q$ and $D(p,q) = \downarrow q - \downarrow p$. Then $\mathcal{M}$ gives rise to a \emph{partial resource monoid}, 
namely a partial monoid that satisfies all the axioms of resource monoids given by 
\cite{dal-semantic}. 

It would be desirable to have a \emph{total} resource monoid
so that 
the basic results of \cite{dal-semantic} would be reused for our purpose.
However, we have no idea how to do that coherently. This problem is related to
the above mentioned restriction on $\DIAL$
that the premise of a linear implication 
must be a linear formula.
\end{remark}

\subsection{Realizability relation}

We are now ready to introduce the realizability relation. 
Intuitively, $t,p \Vdash A$ signifies that $A$ is the specification of $t$ and
$p$ majorizes the potential cost for evaluating $t$ when 
it is applied to some arguments.

Let us begin with some notations.
\begin{itemize}
\item $\overline{x}$, $\overline{t}$, $\overline{A}$ stand for
(possibly empty) lists of variables, terms and formulas, respectively. 
\item 
$t \langle u_1/x_1, \dots, u_n/x_n\rangle$ denotes the term
$(\lambda x_1\cdots x_n. t)u_1 \cdots u_n$.
\item  $\theta, \xi$ stand for lists 
of binding expressions; for instance, 
$\theta  = u_1/x_1, \dots, u_n/x_n$ with $x_1, \dots, x_n$ distinct. 
This allows us to concisely write 
$t \langle \theta\rangle$ for
$t \langle u_1/x_1, \dots, u_n/x_n\rangle
= (\lambda x_1\cdots x_n. t)u_1 \cdots u_n$.
\end{itemize}

\begin{definition}{(Saturated sets)}
Let $\tau$ be a type for higher order polynomials.
A nonempty set $X \subseteq \Lambda \times \Pi$ is 
a \emph{saturated set of type} $\tau$ if
whenever $(t,p) \in X$, we have $t\Downarrow$,
$p$ is a closed higher order polynomial 
of type $\tau$ and
the following hold:
\begin{description}
\item[(bound)] $TS(t) \leq \downarrow p$.
\item[(monotonicity)] $(t,\; p+n)\in X$ for every $n\in \N$.
\item[(exchange)] If $t= t_0\langle \theta,v_1/y_1, v_2/y_2, \xi\rangle\vu$,
then $(t_0\langle \theta,v_2/y_2, v_1/y_1, \xi\rangle\vu,\; p)\in X$.
\item[(weakening)] If $t= t_0\langle \theta\rangle\vu$, $z\not\in FV(t_0)$
and $w\Downarrow$,
then $(t_0\langle \theta,w/z\rangle\vu,\; p + TS(w)+2)\in X$.
\item[(contraction)] If $t= t_0\langle \theta,w/z_1, w/z_2\rangle\vu$,
then $(t_0[z/z_1, z/z_2]\langle \theta,w/z\rangle\vu,\; p)\in X$.
\item[(concatenation)] If $t= (t_0\langle \theta\rangle)
(t_1\langle \xi\rangle)\vu$, then 
$((t_0t_1)\langle\theta,\xi\rangle\vu,\; p)\in X$.
\item[(identity)] If $t= t_0 \vu$, then 
$((x\langle t_0/x\rangle)\vu,\; p+3) \in X$.
\end{description}
\end{definition}

By Lemma \ref{value_lemma} (size), condition (bound) implies 
that $|\Val{t}| \leq \downarrow p$.
Note that condition (weakening) asks for 
an additional cost $TS(w)+ 2$.
This is due to our computational model: weak call-by-value reduction
$(\lambda x.t)w \rightarrow t[w/x]$ requires that $w$ is a value, even when
$x\not \in FV(t)$.

We have to show that there exists at least one saturated set.
The following proposition gives the canonical one.
\begin{prop}
$X_0 = \{ (t,n) : t\Downarrow \mbox{ and } TS(t) \leq n\}$ 
is the greatest saturated set of type $o$.
\end{prop}
\begin{proof}
Conditions (bound) and (monotonicity) hold by definition.
The other conditions follow from Lemma \ref{value_lemma}.
$X_0$ is obviously greatest.
\end{proof}

A \emph{valuation} $\eta$ maps each propositional variable $\alpha$ 
to a saturated set $\eta(\alpha)$
of type $o$. 
$\eta\{\alpha \leftarrow X\}$ stands for 
a valuation which agrees with $\eta$ except that it  assigns $X$
to $\alpha$.

\begin{definition}{(Realizability)}
\label{real_def}
We define the relation $t,p \Vdash_\eta A$, 
where $t\in \Lambda$ (called \emph{realizer}), 
$p$ is a closed higher order polynomial of type $o(A)$ (called \emph{majorizer})
and $\eta$ is a valuation. 
It induces the set $\hat{A}_\eta =
\{ (t,p) : t,p\Vdash_\eta A\}$.
The definition proceeds by induction on $A$.
\begin{itemize}
\item[$\bullet$] $t,\,n\Vdash_\eta \alpha$ iff $(t,\,n)\in \eta(\alpha)$.
\item[$\bullet$] $t,\,p\Vdash_\eta L \multimap A$ 
iff $TS(t)\leq \downarrow p$ and 
$u, m \Vdash_\eta L$ implies $tu, p+m \Vdash_\eta A$ for every $u, m$.
\item[$\bullet$] $t,\, p \Vdash_\eta B \Rightarrow A$ 
iff $TS(t)\leq \downarrow p$ and 
$u, q \Vdash_\eta B$ implies $tu, p(q) \Vdash_\eta A$
for every $u, q$.
\item[$\bullet$] $t,\, p \Vdash_\eta \forall \alpha A$ iff 
$t,\, p \Vdash_{\eta\{\alpha \leftarrow X\}} A$
for every saturated set $X$ of type $o$.
\item[$\bullet$] $t,\, p \Vdash_\eta \mu \alpha L$ iff 
$(t,p) \in X$ for every saturated set $X$ of type $o$ 
such that $\hat{L}_{\eta\{\alpha \leftarrow X\}} \subseteq X$.
\end{itemize}
\end{definition}

\begin{lemma}\label{l-real}~
\begin{enumerate}
\item For every formula $A$,
$\hat{A}_\eta = \{ (t,p) : t,p\Vdash_\eta A\}$ is a saturated set of 
type $o(A)$.

\item For every $A$ and $L$, we have
  $t,p \Vdash_\eta A[L / \beta]$  iff  
$t,p \Vdash_{\eta\{\beta \leftarrow \hat{L}_\eta\}} A$.

\item If $t, p \Vdash_\eta \forall \alpha A$, then
$t, p \Vdash_\eta  A[L/\alpha]$ for every linear formula $L$.

\item $\hat{\mu\alpha L}_\eta$ is the least fixpoint of $f(X) = 
\hat{L}_{\eta\{\alpha \leftarrow X\}}$.

\item $t, p \Vdash_\eta \mu \alpha L$ iff
$t, p \Vdash_\eta  L[\mu\alpha L/\alpha]$.
\end{enumerate}
\end{lemma}

\begin{proof}
1 and 2. By induction on $A$, noting that any non-empty 
intersection of
saturated sets is again saturated.\\
3. By 1 and 2.\\
4. Notice that $f$ is a monotone function since 
$\alpha$ occurs only positively in $L$. 
Call a saturated set $X$ of type $o$ a \emph{prefixpoint} of $f$ 
if $f(X)\subseteq X$. Then $\hat{\mu\alpha L}_\eta$
is the infimum of all prefixpoints of $f$. So 
$\hat{\mu\alpha L}_\eta \subseteq X$ for every prefixpoint $X$, 
and by monotonicity
$f(\hat{\mu\alpha L}_\eta) \subseteq f(X) \subseteq X$.
Since $\hat{\mu\alpha L}_\eta$ is the infimum of all such $X$'s, 
we obtain 
$f(\hat{\mu\alpha L}_\eta) \subseteq \hat{\mu\alpha L}_\eta$.
Applying monotonicity again, we get
$f(f(\hat{\mu\alpha L}_\eta)) \subseteq f(\hat{\mu\alpha L}_\eta)$,
so $f(\hat{\mu\alpha L}_\eta)$ is a prefixpoint of $f$. 
Hence 
$\hat{\mu\alpha L}_\eta \subseteq 
f(\hat{\mu\alpha L}_\eta)$. 

5. By 2 and 4, $t, p \Vdash_\eta \mu \alpha L$ iff
$(t,p) \in \hat{\mu\alpha L}_\eta$ iff
$(t,p) \in f(\hat{\mu\alpha L}_\eta$) iff
$t, p \Vdash_{\eta\{\alpha \leftarrow \hat{\mu\alpha L}_\eta\}} L$
iff $t,p \Vdash_\eta L[\mu\alpha L/\alpha]$.
\end{proof}

\subsection{Adequacy theorem}

The adequacy theorem is the crux of this paper. 
It states that $\DIAL$ is sound with respect to the 
realizability semantics we have introduced.

\begin{theorem}[Adequacy]
\label{adequacy_lemma}
Suppose that $\vx:\vC; \vy:\vM \vdash t: A$ is derivable.
Then there exists a higher order polynomial $p(\vx): o(A)$ 
with variables $\vx$ of type $o(\vC)$ such that 
for any valuation $\eta$ we have the following:
$$\vu, \vq\Vdash_\eta \vC,\ \vs, \vm\Vdash_\eta \vM
\Longrightarrow 
t\langle\vu/\vx, \vs/\vy\rangle, p(\vq)+\vm \Vdash_\eta A.$$
Moreover, if $\vx = x_1, \dots, x_a$
and each $x_i$ occurs $c_i$ times in $t$, then
$$
(*)\quad |t| + c_1 \downarrow q_1 + \cdots + 
c_a \downarrow q_a \leq \downarrow p(\vq).$$
\end{theorem}

We call the above $p(\vx)$ a {\em majorizer} of 
$\vx:\vC; \vy:\vM \vdash t: A$. 

\begin{proof}
By induction on the length of the derivation. 
We omit the cases for $\forall$ and $\mu$, 
since they easily follow from Lemma \ref{l-real}. 
Accordingly, 
we do not specify the valuation $\eta$,
simply writing $\Vdash$ for $\Vdash_\eta$.
We distinguish the last inference rule of the derivation.

\medskip\noindent
Case (ax1): For $x : A ; \vdash x : A$, take $p(x) = x + 3$ as 
the majorizer. Condition (*) obviously holds.

If $u, q \Vdash A$, then condition (identity) for saturated sets implies
$x\langle u/x\rangle, q+3 \Vdash A$, namely 
$x\langle u/x\rangle, p(q) \Vdash A$.

\medskip\noindent
Case (ax2): For $; y : L \vdash y : L$, take $p = 3$ as 
the majorizer.

\medskip\noindent
Case ($\multimap_e$): \qquad
\raisebox{-1em}{
$\infer{\vx_1:\vC_1, \vx_2:\vC_2; \vy_1:\vM_1, \vy_2:\vM_2 \vdash t_1 t_2: A}{
\vx_1:\vC_1; \vy_1:\vM_1 \vdash t_1: L\multimap A
&
\vx_2:\vC_2; \vy_2:\vM_2 \vdash t_2: L}$}

\medskip\noindent
By the induction hypothesis, we have majorizers 
$p_1(\vx_1): o(L\multimap A)$ and $p_2(\vx_2): o(L)$
of the left and right premises, respectively.
We claim that $p(\vx_1,\vx_2) = p_1(\vx_1)+ p_2(\vx_2)$ 
is the suitable 
majorizer of the conclusion. Notice that $p_2(\vx_2)$ is of type $o$,
so that the addition is well defined. Condition (*) follows
by the induction hypothesis.

Suppose that $\vu_i,\vq_i\Vdash \vC_i$ and
$\vs_i,\vm_i\Vdash \vM_i$ for $i=1,2$
and write  $\theta_i$ for the list
$\vu_i/\vx_i, \vs_i/\vy_i$.

Then the induction hypothesis yields 
$t_1\langle \theta_1\rangle, p_1(\vq_1)+\vm_1 \Vdash L\multimap A$
and 
$t_2\langle\theta_2\rangle, p_2(\vq_2)+\vm_2 \Vdash L$.
Hence by the definition of realizability,
$t_1\langle \theta_1\rangle
t_2\langle\theta_2\rangle,
p_1(\vq_1)+p_2(\vq_2)+\vm_1+\vm_2 \Vdash A$, 
so by conditions (concatenation) and (exchange),
$(t_1t_2)[\vu_1/\vx_1, \vu_2/\vx_2, \vs_1/\vy_1,\vs_2/\vy_2], 
p(\vq_1,\vq_2)+\vm_1 + \vm_2 \Vdash A$
as required.

\medskip\noindent
Case ($\Rightarrow_e$): \qquad
\raisebox{-1em}{
$
\infer{\vx_1:\vC_1,\vx_2:\vC_2; \vy:\vM \vdash t_1 t_2: A}{
\vx_1 : \vC_1 ; \vy : \vM \vdash t_1: B\Rightarrow A
&
\vx_2 : \vC_2 ; \vdash t_2 : B}
$}

\medskip\noindent
By the induction hypothesis, we have majorizers 
$\lambda z.p_1(\vx_1,z): o(B)\rightarrow o(A)$ and $p_2(\vx_2): o(B)$
of the left and right premises, respectively.
We claim that $p(\vx_1,\vx_2) = p_1(\vx_1, p_2(\vx_2))$ is the suitable 
majorizer of the conclusion. As before, condition (*) holds.

Suppose that $\vu_i,\vq_i\Vdash \vC_i$ for $i=1,2$,
$\vs,\vm\Vdash \vM$, and write $\theta_1 = \vu_1/\vx_1, \vs/\vy$
and $\theta_2 = \vu_2/\vx_2$.
Then the induction hypothesis yields
$t_1\langle\theta_1\rangle, \lambda z. p_1(\vq_1,z)+\vm \Vdash B\Rightarrow A$
and 
$t_2\langle\theta_2\rangle, p_2(\vq_2) \Vdash B$. 
Hence
$t_1\langle\theta_1\rangle
t_2\langle \theta_2\rangle,$
$p_1(\vq_1,p_2(\vq_2))+\vm \Vdash A$, so
by conditions (concatenation) and (exchange),
$(t_1 t_2)\langle\vu_1/\vx_1, \vu_2/\vx_2, \vs/\vy\rangle),$
$p(\vq_1,\vq_2)+\vm \Vdash A$
as required.

\medskip\noindent
Case ($\multimap_i$): \qquad
\raisebox{-1em}{
$
\infer{\vx:\vC; \vy:\vM \vdash \lambda z.t: L\multimap A}{
\vx:\vC; \vy:\vM, z:L \vdash t: A}
$
}

\medskip\noindent
By the induction hypothesis, we have a majorizer
$p_0(\vx): o(A)$
of the premise. 
We claim that 
$p(\vx)= p_0(\vx)+ d$ with constant $d$ specified below
 is the suitable majorizer
of the conclusion. Condition (*) holds if $d\geq 1$.

Suppose that $\vu,\vq\Vdash \vC$, $\vs,\vm\Vdash \vM$ 
and write $\theta = \vu/\vx,\vs/\vy$.
Then whenever $w,k\Vdash L$, the induction hypothesis gives us
$t\langle\theta, w/z\rangle, p_0(\vq)+\vm + k\Vdash A$.
By (monotonicity),
$t\langle \theta, w/z\rangle, p(\vq)+\vm + k\Vdash A$,
namely, $w,k\Vdash L$ implies 
$(\lambda z.t)
\langle \theta\rangle w,\; p(\vq)+\vm + k\Vdash A$.

Hence it just remains to verify that 
$((\lambda z.t)
\langle \theta\rangle,\; p(\vq)+\vm)$ satisfies condition (bound).
Suppose that $\vx = x_1, \dots, x_a$
and
each $x_i$ occurs at most $c_i$ times in $t$.
We assume that $c_i \geq 2$ for $i=1, \dots, a$;
the case $c_i=0,1$ can be easily treated by choosing $d$ large 
enough. 
We have
$$(\lambda z.t)\langle \theta\rangle 
= 
(\lambda \vx\vy z.t)\vu\vs
\stackrel{n_1}{\Arr}
(\lambda \vx\vy z.t)\Val{\vu}\Val{\vs}
\stackrel{n_2}{\Arr}
(\lambda \vy z.t[\Val{\vu}/\vx])\Val{\vs}
\stackrel{n_3}{\Arr}
\lambda z.t[\Val{\vu}/\vx,\Val{\vs}/\vy],$$
where
$n_1 = Time(\vu)+ Time(\vs)$, 
$n_3 \leq d_1 =$ the length of the list $\vy$, and
$n_2  =  |\lambda \vy z.t[\Val{\vu}/\vx]| - | 
(\lambda \vx\vy z.t)\Val{\vu}|$ 
$  \leq  (c_1 -1)|\Val{u_1}| +
\cdots + (c_a -1)|\Val{u_a}|
\leq 
(c_1 -1)TS(u_1) +
\cdots + (c_a -1)TS(u_a)$
by Lemmas \ref{l-cost} and \ref{value_lemma}(1).
Hence
\begin{eqnarray*}
TS((\lambda z.t)\langle \theta\rangle) & = &
Time((\lambda \vx\vy z.t)\vu\vs) + |(\lambda \vx\vy z.t)\vu\vs| \\
& \leq & Time(\vu) + Time(\vs) +
(c_1 -1)TS(u_1) +
\cdots + (c_a -1)TS(u_a) + d_1 + |t| + |\vu| + |\vs| + d_2\\
& = & |t|+ c_1 TS(u_1) + \cdots + c_a TS(u_a) + TS(\vs) + d_1 + d_2,
\end{eqnarray*}
where $d_2$ is the length of $\vx\vy z$.
Because of $TS(\vu) \leq \downarrow\vq$, $TS(\vs)\leq \vm$ 
and condition (*), we obtain
$TS((\lambda z.t)\langle \theta\rangle)$ $\leq p(\vq) + \vm$
by letting $d= d_1 + d_2$. We therefore conclude 
$(\lambda z.t)
\langle \theta\rangle,\; p(\vq)+\vm\Vdash L\multimap A$.

\medskip\noindent
Case ($\Rightarrow_i$): \qquad
\raisebox{-1em}{
$
\infer{\vx:\vC; \vy:\vM \vdash \lambda z.t: B\Rightarrow A}{
\vx:\vC, z:B; \vy:\vM \vdash t: A}
$
}

\medskip\noindent
By the induction hypothesis, we have a majorizer
$p_0(\vx,z): o(A)$ of the premise.
We claim that 
$p(\vx)= \lambda z. p_0(\vx,z)+ d: o(B)\rightarrow o(A)$ 
with $d$ a large enough constant
is the suitable majorizer
of the conclusion.
The proof is just the same as above.

\medskip\noindent
Case ($Contr$): \qquad
\raisebox{-1em}{
$
\infer{ z:B, \vx:\vC; \vy:\vM \vdash t[z/z_1,z/z_2]: A}{
z_1:B, z_2: B, \vx:\vC ; \vy:\vM \vdash t: A}
$}

\medskip\noindent
By the induction hypothesis, we have a majorizer
$p_0(z_1,z_2,\vx): o(A)$
of the premise. We can prove that 
$p(z,\vx)= p_0(z,z,\vx)$
is the suitable majorizer of the conclusion by using
condition (contraction).

\medskip\noindent
Case ($Weak$): \qquad
\raisebox{-1em}{
$
\infer{\vx:\vC; \vy:\vM, z:L \vdash t: A}{
\vx:\vC; \vy:\vM \vdash t: A}
$
}

\medskip\noindent
By the induction hypothesis, we have a majorizer
$p_0(\vx): o(A)$
of the premise. 
We can prove that 
$p_0(\vx)+2: o(A)$ works for the conclusion
by using condition (weakening).\\

\medskip\noindent
Case ($Derel$): \qquad
\raisebox{-1em}{
$
\infer{\vx:\vC, z:L; \vy:\vM \vdash t: A}{
\vx:\vC; \vy:\vM, z:L \vdash t: A}
$}

\medskip\noindent
By the induction hypothesis, we have a majorizer
$p_0(\vx): o(A)$
of the premise. Then it is easy to see that
$p(\vx,z)= p_0(\vx)+z$ works as a majorizer of the conclusion.
\end{proof}

\section{Polynomial time soundness}
\label{section_4}

In this section, we apply the adequacy theorem to prove
that every term $t$ of type $\church{W}\Rightarrow L$ (with $L$ a 
Scott data type) represents a polynomial time function. 
There is, however, a technical problem due to the use of
the weak call-by-value strategy.
Since it does not reduce under $\lambda$,
if $t$ is of the form 
$t = \lambda xy.t'$, then the evaluation of 
$t\,\church{w}$ gets stuck after the first reduction.

The problem can be settled by a little trick when 
$L$ is fixpoint-free,
eg., $L = \scott{B}_2$ (Subsection \ref{ss-predicate}).
However, the case when $L = \scott{W}$ is not so easy.
The main difficulty is that,
although each \emph{bit} of the output Scott word 
can be computed in polynomial time,
its \emph{length} is not yet ensured to be polynomial. 
The length cannot be detected by 
weak call-by-value; it rather 
depends on the size of the $\beta$-normal form.
We are thus compelled to develop another realizability argument 
based on the $\beta$-normal form,
which indeed ensures that the output is of polynomial length
(Subsection \ref{ss-size}). We will then be able to prove the polynomial time
soundness for the Scott words (Subsection \ref{ss-word}).



\subsection{Polynomial time soundness for predicates}
\label{ss-predicate}
We first observe that 
Church numerals and words are bounded by linear majorizers.

\begin{lemma}~
\begin{enumerate}
\item For every $n\in\N$, we have
$\church{n}, p_n\Vdash \church{N}$
with $p_n = \lambda z. n(z + 3)+ 3: o\rightarrow o$. 
\item For every $w\in\{0,1\}^n$, we have 
$\church{w}, q_n \Vdash \church{W}$
with $q_n = \lambda z_0 z_1. n(z_0 + z_1 + 3)+ 3 
: o\rightarrow o\rightarrow o$.
\end{enumerate}
\end{lemma}

\begin{proof}
Since both are similar, we only prove the statement 1.
We assume $n\geq 1$,
since the case $n=0$ is easy.
Let $\eta$ be a valuation, 
$u, m\Vdash_\eta \alpha\multimap \alpha$ and
$v, k\Vdash_\eta \alpha$.

By condition (identity), we have
$x\langle v/x\rangle, k+3\Vdash_\eta \alpha$ and
$f_i\langle u/f_i\rangle, m+3\Vdash_\eta \alpha\multimap \alpha$ 
for any variable $f_i$.
So 
$f_1 \langle u/f_1 \rangle(\cdots (f_n\langle u/f_n\rangle
 x\langle v/x))\cdots), n(m+3)+ k+3 \Vdash_\eta \alpha$.
By (concatenation) and (contraction),
$f^n x \langle u/f,v/x\rangle, n(m+3)+ k+3 
\Vdash_\eta \alpha$.
By noting that $f^n x \langle u/f,v/x\rangle = (\lambda fx.f^n x)uv$, 
we obtain $\lambda fx.f^n x,$ $\lambda z. n(z+3)+3
\Vdash_\eta (\alpha\multimap\alpha)\Rightarrow
(\alpha\multimap\alpha)$.
\end{proof}
	  	  
These linear majorizers are
turned into polynomial ones 
when applied to majorizers of higher order type.
As a consequence, we obtain a polynomial bound on the execution time.

\begin{theorem}[Weak soundness]
\label{weak_sound}
Let $L$ be a linear formula.
If $\vdash t : \church{W} \Rightarrow L$, then 
there exists a polynomial $P$ such that for every $w\in\{0,1\}^*$,
$Time(t\church{w}) \leq P(|w|)$.
\end{theorem}
\begin{proof}
By the adequacy theorem, 
we have a majorizer $\lambda x.p(x) : o(\scott{W}\Rightarrow L) 
= (o \rightarrow o \rightarrow o) \rightarrow o$ 
such that $t,\, \lambda x.p(x) \Vdash \church{W} \Rightarrow L$. 
Let $w\in\{0,1\}^n$.
By the lemma above, we have $\church{w},q_n \Vdash_\eta W$. 
Hence
by the definition of realizability, 
$t\church{w},\, p(q_n) \Vdash L$. 
	
We prove that $p(q_n):o$ is a polynomial in $n$ by induction 
on the structure of the term $p(x)$. 
We suppose that $p(x)$ is in $\beta$ normal form. 

If $p(x) = k$, then $p(q_n)=k$ is a constant and obviously a polynomial in $n$.
If $p(x) = p_1(x) + p_2(x)$, then
by the induction hypothesis $p_1(q_n)$ and $p_2(q_n)$ are
polynomials in $n$, so is $p(q_n)$. Otherwise, $p(x)$ must be 
of the form $x p_1(x)p_2(x)$ since $x$ is the only free variable and
of type $o\rightarrow o\rightarrow o$.
By the induction hypothesis, $p_1(q_n)$ and $p_2(q_n)$ are polynomials in $n$. 
So $p(q_n) = q_n(p_1(q_n),p_2(q_n)) = n(p_1(q_n) + p_2(q_n)+3) + 3$ is still a polynomial in $n$.

By condition (bound), we conclude that 
$Time(t\church{w})$ is bounded by $p(q_n)$, a polynomial in $n$.
\end{proof}

This in particular implies that every term of type 
$\church{W}\Rightarrow \scott{B}_2$ 
represents a polynomial time predicate.

\begin{corollary}[\class{P}-soundness for predicates]
\label{c-predicate}
If $t: \church{W} \Rightarrow \scott{B}_2$, then the predicate
 $f_t : \{0,1\}^* \rightarrow \{0,1\}$ 
defined by $f_t(w) = 1$ $\Leftrightarrow$ 
$\Val{t\church{w}}_\beta =\scott{b_1}$
is a polynomial time predicate.
\end{corollary}
\begin{proof}
Observe that 
$\lambda x. tx\scott{b}_0\scott{b}_1: \church{W} 
\Rightarrow \scott{B}_2$ and for every $w\in\{0,1\}^*$
the term $(\lambda x. tx\scott{b}_0\scott{b}_1)\church{w}$
reduces to either $\scott{b}_0$ or $\scott{b}_1$ 
by the weak call-by-value strategy
(see the proof of Lemma \ref{l-trick}). 
By the previous theorem, the runtime is 
bounded by a polynomial.
\end{proof}

\subsection{Size realizability}
\label{ss-size}
As explained in the beginning of this section, the previous realizability
semantics does not tell anything about the length of the output 
Scott words.
We thus introduce another realizability semantics based on 
the (applicative) size of $\beta$-normal forms.
Due to lack of space, we can only state the definitions and the result.

Let $\sharp t$ be the number of applications in $t$,
which is more precisely defined by:
$$
\sharp x = 0, \qquad
\sharp (tu)  =  \sharp t + \sharp u + 1, \qquad
\sharp \lambda x.t  =  \sharp t.
$$
$\sharp t$ is not relevant for bounding the size of $t$ in general
(think of $t= \lambda x_1 \cdots x_{100} . x_i$; we have 
$\sharp t = 0$). However, when $t$ is a Scott word,
$\sharp t$ exactly corresponds to the length of 
the word represented by $t$.

\begin{definition}[Size-saturated sets]
Let $\tau$ be a type for higher order polynomials.
A nonempty set $X \subseteq \Lambda \times \Pi$ is 
\emph a {size-saturated set of type} $\tau$ if
whenever $(t,p) \in X$, $t$ is normalizable, 
$p$ is a closed higher order polynomial 
of type $\tau$ and
the following hold:
\begin{description}
\item[(bound')] $\sharp \Val{t}_\beta \leq \downarrow p$.
\item[(weak')] if $t= t_0\langle \theta\rangle\vu$ and $z\not\in FV(t_0)$,
then $(t_0\langle \theta,w/z\rangle\vu,\; p)\in X$. 
\item[(identity')] if $t= t_0 \vu$, then 
$((x\langle t_0/x\rangle)\vu,\; p) \in X$.
\end{description}
We also require conditions (monotonicity), (exchange), (contraction),
(concatenation) of Definition \ref{real_def}, and finally,
\begin{description}
\item[(variable)] $(\maltese, 0_\tau)\in X$,
where $\maltese$ is a fixed variable.
\end{description}
\end{definition}

Condition (variable) employs a fixed variable
$\maltese$ (considered as an inert object), 
that helps us to deal with open terms. 
Notice that it
contradicts the previous condition (bound); that is one reason
why we have to 
consider size realizability separately from the previous
one.

As before, we have the greatest size-saturated set of type $o$:
$X_s = \{ (t, n) : \sharp \Val{t}_\beta \leq  n\}$.
A \emph{valuation} $\eta$ is now supposed to map
each propositional variable to a size-saturated set of type $o$.

\begin{definition}[Size realizability]
We define the relation $t,p \Vdash^s_\eta A$ 
as in Definition \ref{real_def}, except that
 \begin{itemize}
	\item[$\bullet$] $t,\,p\Vdash^s_\eta L \multimap A$
iff either $t=\maltese$ (and $p$ is arbitrary), or 
$u, m \Vdash^s_\eta L$ implies $tu, p+m \Vdash^s_\eta A$ for every $u, m$.
	\item[$\bullet$] $t,\, p \Vdash^s_\eta B \Rightarrow A$
iff either $t=\maltese$, or 
$u, q \Vdash^s_\eta B$ implies $tu, p(q) \Vdash^s_\eta A$
for every $u, q$.
\end{itemize}
\end{definition}

One can then verify that 
for every formula $A$ and valuation $\eta$, the set 
$\hat{A}_\eta = \{ (t,p) : t,p\Vdash_\eta A\}$ is a 
size-saturated set of type $o(A)$.

\begin{theorem}[Size adequacy]
Suppose that $\vx:\vC; \vy:\vM \vdash t: A$ is derivable.
Then there exists a higher order polynomial $p(\vx): o(A)$ such that 
for any valuation $\eta$ we have the following:
$$\vu, \vq\Vdash^s_\eta \vC,\ \vv, \vm\Vdash^s_\eta \vM
\Longrightarrow 
t\langle\vu/\vx, \vv/\vy\rangle, p(\vq)+\vm \Vdash^s_\eta A.$$
\end{theorem}

\begin{theorem}[Size soundness]
\label{size_sound}
	If $\vdash t : \church{W} \Rightarrow L$, then 
there exists a polynomial $P$ such that
	$\sharp \Val{t\church{w}}_\beta \leq P(|w|)$
for every $w\in\{0,1\}^*$.
\end{theorem}


\subsection{Polynomial time soundness for words}
\label{ss-word}
As in the proof of Corollary \ref{c-predicate}, 
we use a little trick. First note that 
we have a predecessor 
$\scott{p} = \lambda z.z(\lambda x.x)(\lambda x.x)(\scott{\epsilon}): 
\scott{W} \multimap \scott{W}$.
By employing it, we define 
$$
\scott{q}  =  \lambda x. x (\lambda y.\scott{b}_0)
(\lambda y. \scott{b}_1) \scott{b}_2 : 
\scott{W} \multimap \scott{B}_3, \qquad
\scott{bit}_i  =  \lambda x.\scott{q}(
\underbrace{\scott{p}\cdots\scott{p}}_{i\ times}(x)) :\scott{W} 
\multimap \scott{B}_3.
$$
\begin{lemma}\label{l-trick}
Suppose that $t$ is a closed term of type $\scott{W}$ and 
$\Val{t}_\beta$ represents a word $w\in \{0,1\}^n$ of length $n$. 
Then for any $i < n$,
$\Val{\scott{bit}_i(t)} = \scott{b}_0$ or $\scott{b}_1$, depending on
the $i$th bit of $w$. If $i\geq n$,
$\Val{\scott{bit}_i(t)} = \scott{b}_2$.
\end{lemma}
\begin{proof}
The crucial fact
is that given a closed term $u$ of type $\scott{W}$,
$\scott{q} u$ always evaluates to 
$\scott{b}_j$ for some $j \in \{0,1,2\}$ by the weak call-by-value
strategy. To see this, take a fresh propositional variable $\gamma$,
variables $z_0, z_1, z_2$, and consider
$\scott{q}_\gamma  =  \lambda x. x (\lambda y.z_0)
(\lambda y. z_1) z_2: \scott{W} \multimap \gamma$. Since
$\scott{q}_\gamma u$ is of type $\gamma$, so is
$\Val{\scott{q}_\gamma u}$ by the subject reduction property.
Hence it cannot be an abstraction. It cannot either be an application,
since the only possible 
head variables are $z_0, z_1$ and $z_2$ of atomic type $\gamma$.
Therefore $\Val{\scott{q}_\gamma u}= z_j$ for some $j \in \{0,1,2\}$.
By substituting $\scott{b}_j$ for $z_j$, we obtain
$\Val{\scott{q}u}= \scott{b}_j$. Now the claim is easily verified.
\end{proof}

Finally we are able to prove the polynomial time soundness for words.

\def\thetheonum{\ref{soundness_theo}}
\begin{theorm_num}{\class{FP}-soundness}
For every $\lambda$-term $t$ of type $\church{W} \Rightarrow \scott{W}$, 
the associated function 
$f_t : \{0,1\}^* \rightarrow \{0,1\}^*$ defined by $f_t(w_1) = w_2$ 
$\Leftrightarrow$ $\Val{ t \church{w_1}}_\beta = \scott{w_2}$
 is a polynomial time function.
\end{theorm_num}

\begin{proof}
By the adequacy theorem, we have 
$t,\; \lambda x.p(x)\Vdash \church{W} \Rightarrow \scott{W}$
for some $\lambda x.p(x): o(\church{W})\rightarrow o$.

We also have $\scott{p}, k\Vdash \scott{W} \multimap \scott{W}$
and $\scott{q}, k'\Vdash \scott{W} \multimap \scott{B}_3$
for some constants $k$, $k'$, from which we easily obtain
$\lambda x.\scott{bit}_i(tx),\; \lambda x.p(x) + ik + k'' \vdash 
\church{W} \Rightarrow \scott{B}_3$ for some constant $k''\geq k'$.

By inspecting the proof of Theorem \ref{weak_sound} we obtain
a polynomial $P(x)$ such that 
$Time(\scott{bit}_i(t\church{w}))\leq P(|w|) + ik$ for 
every $w\in \{0,1\}^n$ and every $i\in \N$.
Furthermore, Theorem 
\ref{size_sound} 
gives  a polynomial $Q(x)$ such that
$\sharp \Val{t\church{w}}_\beta \leq Q(|w|)$, that implies that 
the Scott term
$\Val{t\church{w}}_\beta$ represents a word of length at most $Q(|w|)$.

Now the desired word $f_t(w)$ can be obtained by computing
the values of 
$\scott{bit}_0(t\church{w}), 
\scott{bit}_1(t\church{w}), 
\scott{bit}_2(t\church{w})$, \ldots
until we obtain 
$\Val{\scott{bit}_m(t\church{w})}=\scott{b}_2$.
We know that $m \leq Q(|w|)$. Hence the overall runtime 
is $R(|w|)$
with $R(x)=O((P(x)+ Q(x))^4 \cdot Q(x))$
in view of 
Theorem \ref{t-turing}.
\end{proof}

\section{Concluding remarks}
\label{conclusion}

Inspired by \cite{leivant1993lambda},
we have introduced a purely logical system $\DIAL$ 
that captures precisely the class of polynomial time functions.
To prove soundness, we have introduced a simple variant 
of the Hofmann-Dal Lago realizability.
Here is a non-exhaustive list of the remaining open questions related to this work:

\begin{itemize}
\item Can we, instead of using a dual type system, directly deal with the !-connective? For the time being, it seems that it would considerably complicate the definition of the realizability relation.
\item We are compelled to introduce 
two realizability interpretations, one for bounding the
runtime, and the other for bounding the length of the output. Is it possible
to integrate them into one realizability interpretation?
\item Is it possible to relate our definition of realizability 
with the original one \cite{dal2005quantitative} more closely?
We have observed that our higher order polynomials are equipped
with the structure of partial resource monoid (Remark \ref{r-monoid}). Our definition of 
realizability is also derived from their notion of length space. 
Establishing an exact
correspondence is, however, left to the future work.
\item We have adapted the tiered recursion characterization of the \class{PTIME} functions. Can we find a suitable logical system as well corresponding to the tiered recursion characterizations of  \class{PSPACE} and \class{ALOGTIME} 
in \cite{leivant1997ramified}, \cite{leivant1995ramified} and \cite{LM00}? 
\end{itemize}

\bibliographystyle{eptcs}


\end{document}